\documentclass{article}

\usepackage{arxiv}

\usepackage[utf8]{inputenc}
\usepackage[T1]{fontenc}
\usepackage{hyperref}
\usepackage{silence}
\usepackage{url}
\usepackage{booktabs}
\usepackage{amsfonts}
\usepackage{amsmath}
\usepackage{amssymb}
\usepackage{nicefrac}
\usepackage{microtype}
\usepackage{graphicx}
\usepackage{natbib}
\usepackage{doi}
\usepackage{amsthm}

\newtheorem{theorem}{Theorem}[section]
\newtheorem{lemma}[theorem]{Lemma}
\newtheorem{corollary}[theorem]{Corollary}
\newtheorem{proposition}[theorem]{Proposition}

\title{The $\beta$-Bound: Drift Constraints for Gated Quantum Probabilities}

\author{
  Jonathon Sendall\\
  OU Philosophy Department\\
  \texttt{jonathon.sendall@ou.ac.uk}
}

\begin{document}
\maketitle

\begin{abstract}
Quantum mechanics provides extraordinarily accurate probabilistic predictions, yet the framework remains silent on what distinguishes quantum systems from definite measurement outcomes. This paper develops a measurement-theoretic framework for projective gating. The central object is the $\beta$-bound, an inequality that controls how much probability assignments can drift when gating and measurement fail to commute. For a density operator $\rho$, projector $F$, and effect $E$, with gate-passage probability $s = \mathrm{Tr}(\rho F)$ and commutator norm $\varepsilon = \|[F,E]\|$, the symmetric partial-gating drift satisfies $|\Delta p_F(E)| \le 2\sqrt{(1-s)/s} \cdot \varepsilon$. The constant 2 is sharp. We introduce two diagnostic quantities: the coherence witness $W(\rho,F) = \|F\rho(I-F)\|_1$, measuring cross-boundary coherence, and the record fidelity gap $\Delta_T(\rho_F, R)$, measuring expectation-value change under symmetrisation. Three experimental vignettes demonstrate falsifiability: Hong--Ou--Mandel interferometry, atomic energy-basis dephasing, and decoherence-induced classicality. The framework is operationally accessible via multiple routes with platform-dependent trade-offs, and interpretation-neutral, compatible with Everettian, Bohmian, QBist, and collapse approaches. It provides quantitative structure that any interpretation must accommodate, along with a template for experimental tests.
\end{abstract}

\section{Introduction}

\subsection{Motivation}

Quantum mechanics provides extraordinarily accurate probabilistic predictions, yet the framework remains silent on a foundational question: what distinguishes the systems we describe quantum-mechanically from the measurement outcomes we treat as definite? This ``cut'' between quantum and classical domains is typically handled pragmatically \citep{Bohr1935,vonNeumann1955}, but the pragmatic approach obscures structure that may be physically significant.

One natural way to make the cut precise is through projective gating: conditioning on a quantum system passing through a projector that selects for some property, such as lying within a particular subspace or exceeding a coherence threshold defined spectrally. Gating of this kind appears implicitly throughout quantum mechanics, from state preparation to post-selection to the projection postulate itself \citep{Luders1951,BuschEtAl2016}. Yet the quantitative behaviour of probability assignments under gating has received less systematic attention than it deserves.

Standard textbook treatments typically assume that the observable of interest commutes with any conditioning projector, or else treat non-commutativity as a special case requiring separate analysis (e.g., the quantum Zeno effect, weak measurements) \citep{MisraSudarshan1977,AharonovEtAl1988}. This framework takes a different stance: non-commutativity between gate and readout is the generic case, and the resulting probability drift is not an anomaly but a predictable, quantitatively bounded distortion. The $\beta$-bound makes this distortion explicit.

This paper develops a measurement-theoretic framework for projective gating. The central object is the $\beta$-bound, an inequality that controls how much probability assignments can drift when gating and measurement fail to commute. The bound makes precise an intuition that practitioners often invoke informally: that interference between admitted and non-admitted sectors creates predictable distortions, and that these distortions vanish when the gate and the observable are compatible.

We also introduce two diagnostic quantities. The coherence witness measures how much a prepared state carries coherence across the gate boundary. The record fidelity gap measures how much expectation values change when we impose a symmetry that removes certain degrees of freedom. Together with the $\beta$-bound, these quantities form a package that can be tested experimentally.

Our approach is operational rather than interpretive. We do not claim to solve the measurement problem or to adjudicate between interpretations of quantum mechanics. The framework is compatible with Everettian, Bohmian, QBist, and collapse approaches \citep{Everett1957,Bohm1952,FuchsSchack2013,GhirardiRiminiWeber1986}; it provides quantitative structure that any interpretation must accommodate. What we do claim is that projective gating deserves the same careful formal treatment that measurement and dynamics receive, and that this treatment yields falsifiable predictions.

\subsection{Overview of Results}

The main results of this paper are:

\textbf{The $\beta$-bound (Theorem 4.1).} Let $F$ be a projector (the gate), $E$ an effect (the readout), and $\rho$ a density operator (the prepared state). Define the gated probability as $p_\rho(E) = \mathrm{Tr}(\rho FEF)/s$, where $s = \mathrm{Tr}(\rho F)$ is the probability of passing the gate. The symmetric partial-gating drift measures the discrepancy between partially-gated and fully-gated probability assignments:
\[
\Delta p_F(E) = \frac{\mathrm{Tr}(\rho FE) + \mathrm{Tr}(\rho EF)}{s} - 2p_\rho(E).
\]
Then
\[
|\Delta p_F(E)| \le 2\sqrt{\frac{1-s}{s}} \cdot \varepsilon,
\]
where $\varepsilon = \|[F,E]\|$ is the operator norm of the commutator. The constant 2 is sharp and cannot be improved uniformly; saturation is achievable in idealised cases with polar decomposition alignment conditions.

\textbf{The coherence witness (Section 3).} The quantity $W(\rho,F) = \|F\rho(I-F)\|_1$ measures cross-boundary coherence of the state relative to the gate. It vanishes exactly when $\rho$ is block-diagonal with respect to the gate, and satisfies $W(\rho,F) \le \sqrt{s(1-s)}$.

\textbf{The record fidelity gap (Section 5).} For a symmetry group $G$ with associated twirl map $\mathcal{T}$, the gap $\Delta_\mathcal{T}(\rho_F, R) = \mathrm{Tr}(\rho_F R) - \mathrm{Tr}(\mathcal{T}[\rho_F] R)$ measures how much a readout's expectation value changes under symmetrisation. This gap is nonzero precisely when the readout is sensitive to degrees of freedom that the twirl removes.

\textbf{Experimental vignettes (Section 6).} We provide three concrete experimental scenarios, specifying preparation, gate, readout, and symmetry in each case, with quantitative predictions that can be tested using standard quantum optical and atomic physics techniques.

\textbf{Operational access (Section 4.6).} We present three complementary routes to determine $\Delta p_F(E)$: (A) tomography and calibration, (B) ancilla-assisted access to $\mathrm{Tr}(\rho\{F,E\})$ (Jordan product), and (C) a gate-reflection identity when the coherent reflection $U_F = 2F - I$ is implementable.

\subsection{Paper Structure}

Section~2 establishes notation and mathematical preliminaries. Section~3 introduces the coherence witness. Section~4 presents the $\beta$-bound, including a worked qubit example. Section~5 develops the record fidelity gap and its connection to symmetrisation. Section~6 presents experimental vignettes. Section~7 discusses the relationship to existing frameworks and the scope of our claims.

\section{Mathematical Preliminaries}

\subsection{Notation and Assumptions}

We work in a Hilbert space $\mathcal{H}$, assumed finite-dimensional throughout the main development. For infinite-dimensional extensions, we assume $\rho$ is trace-class and all gates, effects, and readouts are bounded operators; the results carry over with standard modifications \citep{ReedSimon1980}.

A density operator $\rho$ satisfies $\rho \ge 0$ and $\mathrm{Tr}(\rho) = 1$. A projector $F$ satisfies $F = F^\dagger = F^2$. An effect $E$ satisfies $0 \le E \le I$ \citep{BuschEtAl2016}. A readout operator $R$ is bounded Hermitian with $\|R\| \le 1$; effects are the special case of readouts satisfying $0 \le R \le I$.

We use $\|\cdot\|$ for the operator norm and $\|\cdot\|_1$ for the trace norm. Energy eigenvalues are denoted $\omega_i$ to avoid notational collision with effects $E$.

\subsection{Symbol Table}

\begin{center}
\begin{tabular}{ll}
$s = \mathrm{Tr}(\rho F)$ & probability of passing the gate \\
$\rho_F = F\rho F/s$ & gated (conditioned) state \\
$p_\rho(E) = \mathrm{Tr}(\rho FEF)/s$ & fully gated probability \\
$q_L(E) = \mathrm{Tr}(\rho FE)/s$ & left-partial-gated probability \\
$q_R(E) = \mathrm{Tr}(\rho EF)/s$ & right-partial-gated probability \\
$\Delta p_F(E)$ & symmetric partial-gating drift \\
$W(\rho,F)$ & coherence witness \\
$\varepsilon = \|[F,E]\|$ & gate--readout noncommutativity \\
$\mathcal{T}$ & twirl map (Haar average over symmetry group) \\
$\mathcal{T}^*$ & dual twirl acting on observables \\
\end{tabular}
\end{center}

\subsection{Gated States and Probabilities}

Given a density operator $\rho$ and a projector $F$ with $s = \mathrm{Tr}(\rho F) > 0$, the gated state is
\[
\rho_F = \frac{F\rho F}{s}.
\]
This is the standard L\"uders update for conditioning on the outcome associated with $F$ \citep{Luders1951}. The gated state $\rho_F$ is supported entirely within the range of $F$.

The fully gated probability for an effect $E$ is
\[
p_\rho(E) = \frac{\mathrm{Tr}(\rho FEF)}{s} = \mathrm{Tr}(\rho_F E).
\]
This represents the probability of outcome $E$ given that the system passed the gate $F$.

We also define the one-sided quantities
\[
q_L(E) = \frac{\mathrm{Tr}(\rho FE)}{s}, \quad q_R(E) = \frac{\mathrm{Tr}(\rho EF)}{s}.
\]
These represent expectation values where the gate has been applied on only one side. The symmetric partial-gating drift is
\[
\Delta p_F(E) = q_L(E) + q_R(E) - 2p_\rho(E) = \frac{\mathrm{Tr}(\rho FE) + \mathrm{Tr}(\rho EF)}{s} - \frac{2\mathrm{Tr}(\rho FEF)}{s}.
\]
The drift measures how much the partially-gated quantities deviate from the fully-gated probability.

\textbf{Remark.} For effects $E$, the quantities $q_L(E)$ and $q_R(E)$ are real but need not lie in $[0,1]$ when $[F,E] \ne 0$. They are one-sided conditioned expectation values, not valid probabilities. This is precisely why drift can occur: the failure of $q_L$ and $q_R$ to coincide with the valid probability $p_\rho(E)$ reflects the boundary-crossing structure that the $\beta$-bound controls.

\textbf{Sanity check.} If $[F,E] = 0$, then $FE = EF$ and $FEF = F^2E = FE$. Hence $q_L(E) = \mathrm{Tr}(\rho FE)/s = \mathrm{Tr}(\rho FEF)/s = p_\rho(E)$, and similarly $q_R(E) = p_\rho(E)$, so $\Delta p_F(E) = 0$. The drift vanishes when gate and readout commute.

\subsection{Twirl Maps}

Let $G$ be a compact group acting on $\mathcal{H}$ via unitary representation $g \mapsto U_g$. The twirl map is the Haar average \citep{Watrous2018}
\[
\mathcal{T}[\rho] = \int_G U_g \rho U_g^\dagger \, d\mu(g),
\]
where $\mu$ is the normalised Haar measure. For a finite group, this reduces to
\[
\mathcal{T}[\rho] = \frac{1}{|G|} \sum_{g \in G} U_g \rho U_g^\dagger.
\]
The twirl projects onto the $G$-invariant subalgebra.

The dual twirl acts on observables:
\[
\mathcal{T}^*(R) = \int_G U_g^\dagger R U_g \, d\mu(g).
\]
The key identity relating these is
\[
\mathrm{Tr}(\mathcal{T}[\rho] R) = \mathrm{Tr}(\rho \mathcal{T}^*(R)).
\]
For the $U(1)$ symmetry generated by a Hamiltonian $H$ with spectral decomposition $H = \sum_i \omega_i |i\rangle\langle i|$, the twirl is
\[
\mathcal{T}[\rho] = \frac{1}{2\pi} \int_0^{2\pi} e^{-i\theta H} \rho \, e^{i\theta H} \, d\theta,
\]
which equals dephasing in the energy eigenbasis: $\mathcal{T}[\rho]$ removes all off-diagonal elements $\rho_{ij}$ with $\omega_i \ne \omega_j$ while preserving diagonal elements and coherences within degenerate subspaces.

\section{The Coherence Witness}

\subsection{Definition}

The $\beta$-bound controls drift in terms of the gate--readout commutator, but we also need a diagnostic for the state itself: how much coherence does $\rho$ carry across the gate boundary?

We define the coherence witness for a state $\rho$ relative to a gate $F$:
\[
W(\rho,F) = \|F\rho(I-F)\|_1.
\]
Here $\|\cdot\|_1$ denotes the trace norm. The quantity $F\rho(I-F)$ is the off-diagonal block of $\rho$ in the decomposition $\mathcal{H} = \mathrm{ran}(F) \oplus \mathrm{ran}(I-F)$, and its trace norm measures the total weight of coherence crossing the gate boundary. This construction is related to resource-theoretic measures of coherence \citep{BaumgratzEtAl2014,StreltsovEtAl2017}.

\subsection{Properties}

The coherence witness has several useful properties.

\textbf{Block-diagonal characterisation.} $W(\rho,F) = 0$ if and only if $\rho$ is block-diagonal with respect to the $F/(I-F)$ decomposition. Equivalently, $W(\rho,F) = 0$ if and only if $F\rho(I-F) = 0$, which holds precisely when $\rho$ has no coherence between the admitted sector (range of $F$) and its complement.

\textbf{Symmetry.} Since $\|X\|_1 = \|X^\dagger\|_1$ and $(F\rho(I-F))^\dagger = (I-F)\rho F$, we have
\[
W(\rho,F) = \|F\rho(I-F)\|_1 = \|(I-F)\rho F\|_1.
\]

\textbf{State budget inequality.} The coherence witness is bounded by a function of the gate-passage probability $s = \mathrm{Tr}(\rho F)$.

\begin{lemma}[State budget]
For any density operator $\rho$ and projector $F$ with $s = \mathrm{Tr}(\rho F) \in (0,1]$,
\[
W(\rho,F) \le \sqrt{s(1-s)}.
\]
\end{lemma}

\begin{proof}
The trace norm is characterised by
\[
\|X\|_1 = \sup \{ |\mathrm{Tr}(U^\dagger X)| : \|U\| \le 1 \}.
\]
Let $U$ be any operator with $\|U\| \le 1$. By the Cauchy--Schwarz inequality for the Hilbert--Schmidt inner product,
\[
|\mathrm{Tr}(U^\dagger F\rho(I-F))| = |\mathrm{Tr}((\rho^{1/2}F)^\dagger (\rho^{1/2}(I-F)U))| \le \|\rho^{1/2}F\|_2 \cdot \|\rho^{1/2}(I-F)U\|_2 \le \sqrt{s} \cdot \sqrt{1-s}.
\]
Taking the supremum over all $U$ with $\|U\| \le 1$ gives the result.
\end{proof}

The bound is saturated when $\rho$ is a pure state with equal weight in the admitted and non-admitted sectors: if $\rho = |\psi\rangle\langle\psi|$ with $|\psi\rangle = \sqrt{s} |f\rangle + \sqrt{1-s} |f^\perp\rangle$ for normalised vectors $|f\rangle \in \mathrm{ran}(F)$ and $|f^\perp\rangle \in \mathrm{ran}(I-F)$, then $W(\rho,F) = \sqrt{s(1-s)}$.

\subsection{Diagnostic Role}

The coherence witness tells us whether a prepared state is ``clean'' with respect to a proposed gate. If $W(\rho,F)$ is small, the state is approximately block-diagonal and we expect minimal boundary effects. If $W(\rho,F)$ is large, the state carries substantial cross-boundary coherence and the $\beta$-bound becomes operationally relevant.

In the experimental vignettes of Section~6, the coherence witness serves as a preparation diagnostic: it quantifies how much the prepared state ``sees'' the gate boundary before any measurement is performed. Combined with the gate--readout commutator $\varepsilon = \|[F,E]\|$, which quantifies how much the readout ``sees'' the boundary, the $\beta$-bound controls the resulting drift.

The unified picture is:
\[
\text{drift} \le (\text{state budget}) \times (\text{observable budget}) \times (\text{normalisation cost})
\]
where the state budget is controlled by $W(\rho,F) \le \sqrt{s(1-s)}$, the observable budget is $\varepsilon = \|[F,E]\|$, and the normalisation cost is $1/s$ from conditioning on gate passage. The $\beta$-bound makes this precise.

\section{The \texorpdfstring{$\beta$}{β}-Bound}

\subsection{Statement of the Main Result}

We now state and prove the central inequality of this paper.

\begin{theorem}[$\beta$-bound for partial gating drift]
Let $\rho$ be a density operator, $F$ a projector, and $E$ an effect. Define $s = \mathrm{Tr}(\rho F) \in (0,1]$, $p_\rho(E) = \mathrm{Tr}(\rho FEF)/s$, and $\varepsilon = \|[F,E]\|$. Define the symmetric partial-gating drift
\[
\Delta p_F(E) = \frac{\mathrm{Tr}(\rho FE) + \mathrm{Tr}(\rho EF)}{s} - 2p_\rho(E).
\]
Then
\[
|\Delta p_F(E)| \le 2\sqrt{\frac{1-s}{s}} \cdot \varepsilon.
\]
\end{theorem}

The constant 2 is sharp and cannot be improved uniformly; saturation is achievable in idealised cases with polar decomposition alignment conditions (see Annex A for the saturation criterion). The bound is not merely an upper estimate: for any $\varepsilon > 0$ and any $s \in (0,1)$, there exist $(\rho, F, E)$ achieving $|\Delta p_F(E)|$ arbitrarily close to $2\sqrt{(1-s)/s} \cdot \varepsilon$. This tightness distinguishes the $\beta$-bound from loose worst-case inequalities common in quantum information \citep{NielsenChuang2000}.

\subsection{Proof Sketch and Corollary}

The full proof appears in Annex A. Here we outline the key steps.

\textbf{Step 1: Block decomposition.} Decompose $\mathcal{H} = \mathrm{ran}(F) \oplus \mathrm{ran}(I-F)$ and write operators in $2 \times 2$ block form:
\[
E = \begin{pmatrix} A & B \\ B^\dagger & C \end{pmatrix}, \quad \rho = \begin{pmatrix} \rho_{11} & \rho_{12} \\ \rho_{21} & \rho_{22} \end{pmatrix}
\]
where $B = FE(I-F)$ and $\rho_{21} = (I-F)\rho F$.

\textbf{Step 2: Commutator norm.} For a projector $F$,
\[
[F,E] = \begin{pmatrix} 0 & B \\ -B^\dagger & 0 \end{pmatrix}
\]
so $\|[F,E]\| = \|B\|$. Thus $\varepsilon = \|B\|$.

\textbf{Step 3: Express drift via off-diagonal blocks.} Direct computation gives
\[
\Delta p_F(E) = \frac{\mathrm{Tr}(\rho_{21} B) + \mathrm{Tr}(\rho_{12} B^\dagger)}{s}.
\]

\textbf{Step 4: Apply trace-norm bounds.} Using $|\mathrm{Tr}(XY)| \le \|X\|_1\|Y\|$,
\[
|\Delta p_F(E)| \le \frac{(\|\rho_{21}\|_1 + \|\rho_{12}\|_1)\|B\|}{s}.
\]

\textbf{Step 5: Bound off-diagonal norms.} By Lemma 3.1,
\[
\|\rho_{12}\|_1 \le \sqrt{s(1-s)}, \quad \|\rho_{21}\|_1 \le \sqrt{s(1-s)}.
\]

\textbf{Step 6: Combine.} Substituting,
\[
|\Delta p_F(E)| \le \frac{2\sqrt{s(1-s)} \cdot \|B\|}{s} = 2\sqrt{\frac{1-s}{s}} \cdot \varepsilon. \qed
\]

\begin{corollary}[Witness form of $\beta$-bound]
Under the hypotheses of Theorem 4.1,
\[
|\Delta p_F(E)| \le \frac{2}{s} W(\rho,F) \cdot \|[F,E]\|
\]
and hence, using $W(\rho,F) \le \sqrt{s(1-s)}$,
\[
|\Delta p_F(E)| \le 2\sqrt{\frac{1-s}{s}} \cdot \|[F,E]\|.
\]
\end{corollary}

\subsection{Interpretation}

The $\beta$-bound has a clean physical interpretation. The drift $\Delta p_F(E)$ is nonzero only when the effect $E$ fails to respect the gate $F$, in the sense that $E$ carries amplitude across the $F/(I-F)$ boundary. The commutator norm $\varepsilon = \|[F,E]\|$ measures exactly this failure of compatibility: it equals the norm of the off-diagonal block $FE(I-F)$, which is the component of $E$ that mixes admitted and non-admitted sectors.

The prefactor $\sqrt{(1-s)/s}$ is the unavoidable normalisation penalty when we condition on passing the gate. When $s$ is close to 1 (the state is mostly in the admitted sector), this penalty is small. When $s$ is close to 0 (conditioning on a rare event), the penalty diverges, amplifying any boundary-crossing effects.

The unified picture from Section 3.3 is now precise:
\[
|\Delta p_F(E)| \le 2 \cdot \frac{\sqrt{1-s}}{\sqrt{s}} \cdot \varepsilon = 2 \cdot (\text{state-side factor}) \cdot (\text{observable-side factor})
\]
where the state-side factor $\sqrt{1-s}$ reflects coherence outside the admitted sector, the observable-side factor $\varepsilon$ reflects the readout's boundary-crossing component, and division by $\sqrt{s}$ is the conditioning cost.

\subsection{Worked Example: Qubit with Tilted Gate}

To make the bound concrete, consider a qubit system with computational basis $\{|0\rangle, |1\rangle\}$.

\textbf{Setup.} Let the gate be a rank-1 projector onto a superposition:
\[
F = |f\rangle\langle f|, \quad |f\rangle = \cos\theta |0\rangle + \sin\theta |1\rangle.
\]
Let the readout be the computational-basis effect $E = |0\rangle\langle 0|$. Let the prepared state be $\rho = |\psi\rangle\langle\psi|$ with $|\psi\rangle = \alpha|0\rangle + \beta|1\rangle$, where $|\alpha|^2 + |\beta|^2 = 1$.

\textbf{Gate-passage probability.}
\[
s = \mathrm{Tr}(\rho F) = |\langle f|\psi\rangle|^2 = |\alpha \cos\theta + \beta \sin\theta|^2.
\]

\textbf{Commutator norm.} In the $\{|0\rangle, |1\rangle\}$ basis,
\[
F = \begin{pmatrix} \cos^2\theta & \cos\theta \sin\theta \\ \cos\theta \sin\theta & \sin^2\theta \end{pmatrix}, \quad E = \begin{pmatrix} 1 & 0 \\ 0 & 0 \end{pmatrix}
\]
so
\[
[F,E] = \begin{pmatrix} 0 & -\cos\theta \sin\theta \\ \cos\theta \sin\theta & 0 \end{pmatrix}
\]
and therefore
\[
\varepsilon = \|[F,E]\| = |\cos\theta \sin\theta| = \tfrac{1}{2}|\sin 2\theta|.
\]

\textbf{Bound behaviour.} The $\beta$-bound gives
\[
|\Delta p_F(E)| \le 2\sqrt{\frac{1-s}{s}} \cdot \tfrac{1}{2}|\sin 2\theta| = \sqrt{\frac{1-s}{s}} \cdot |\sin 2\theta|.
\]

Several limiting cases illuminate the bound:

\textbf{Commuting case ($\theta = 0$ or $\theta = \pi/2$).} When $\theta = 0$, $F = |0\rangle\langle 0| = E$, so $[F,E] = 0$ and $\varepsilon = 0$. The bound gives $|\Delta p_F(E)| \le 0$, which is tight. Similarly for $\theta = \pi/2$, where $F = |1\rangle\langle 1|$ and $[F,E] = 0$.

\textbf{Maximal non-commutativity.} When $\theta = \pi/4$, we have $|f\rangle = (|0\rangle + |1\rangle)/\sqrt{2}$ and $\varepsilon = 1/2$. The bound becomes $|\Delta p_F(E)| \le \sqrt{(1-s)/s}$.

\textbf{Rare-event conditioning.} If $|\psi\rangle \approx |f^\perp\rangle$ (the state orthogonal to the gate), then $s \to 0$ and the prefactor $\sqrt{(1-s)/s} \to \infty$. This is the rare-event amplification regime: conditioning on an unlikely outcome magnifies any boundary effects. In practice, preparations ensure $s \ge s_{\min} > 0$, keeping the bound numerically stable.

\subsection{Degenerate-Support Regime}

The $\beta$-bound diverges as $s \to 0$. This is not a defect but a feature: when we condition on an event with vanishing probability, any finite numerator creates an unbounded ratio.

Physically, $s \to 0$ corresponds to rare-event conditioning. The prepared state has almost no weight in the admitted sector, so passing the gate is unlikely, and the post-selection amplifies small effects. This is the same phenomenon that appears in weak-value amplification and post-selected metrology \citep{AharonovEtAl1988,DresselEtAl2014}.

In operational use, one works in regimes where $s$ is bounded away from 0 by preparation design. The vignettes in Section~6 assume $s \ge s_{\min}$ set by experimental parameters, ensuring the bound yields finite, testable predictions.

\subsection{Operationalising the Drift}

The drift $\Delta p_F(E)$ is well-defined for fixed $(\rho,F,E)$, but its experimental determination depends on available control. We distinguish three routes with different burdens and assumptions.

\textbf{Route A (tomography / calibration).} Perform quantum state tomography for $\rho$ and calibrate the POVM element $E$. Then compute $s=\mathrm{Tr}(\rho F)$, $p_\rho(E)=\mathrm{Tr}(\rho FEF)/s$, and the one-sided quantities $q_L(E)$ and $q_R(E)$ from the reconstructed models. This route is the most general but scales poorly with system dimension.

\textbf{Route B (ancilla-assisted Jordan product).} The drift numerator can be written in terms of the Jordan product $\{F,E\}=FE+EF$:
\[
\Delta p_F(E)=\frac{\mathrm{Tr}(\rho\{F,E\})-2\mathrm{Tr}(\rho FEF)}{s}.
\]
One can access $\mathrm{Tr}(\rho\{F,E\})$ without full tomography using an ancilla-assisted interferometric protocol (a Hadamard-test variant) \citep{NielsenChuang2000,Watrous2018}: prepare an ancilla qubit in $|+\rangle$, apply controlled-$F$ and controlled-$E$ operations (with the ancilla as control), then measure the ancilla in the $X$ basis. The resulting statistics estimate the real part of $\mathrm{Tr}(\rho FE)$, and hence $\mathrm{Tr}(\rho\{F,E\})$. This avoids full tomography but requires an ancilla and controlled operations. When $F$ is high-dimensional or realised via post-selection, controlled-$F$ may be comparably demanding to reflection-based access.

\textbf{Route C (gate reflection).} If the platform permits coherent access to both $\mathrm{ran}(F)$ and $\mathrm{ran}(I-F)$, one can implement the unitary reflection
\[
U_F = 2F - I.
\]
In this case the drift admits the identity (derived in the Appendix):
\[
\Delta p_F(E)=\frac{\mathrm{Tr}(\rho E)-\mathrm{Tr}(\rho U_F E U_F)}{2s}.
\]
Operationally this reduces the drift to two expectation values of $E$ (before and after applying $U_F$), together with an estimate of $s$.

\textbf{Feasibility conditions for Route C.} Route C is feasible only when (i) $F$ is a coherently implementable projector in the effective Hilbert space (not realised via irreversible filtering or loss), (ii) the experiment retains coherent access to $\mathrm{ran}(I-F)$ (no irreversible discard of the rejected sector), and (iii) the phase flip on $\mathrm{ran}(I-F)$ can be synthesised with fidelity comparable to other platform operations.

\begin{quote}
\textbf{Remark (Route C trade-off).} The raw contrast $|\mathrm{Tr}(\rho E)-\mathrm{Tr}(\rho U_F E U_F)|$ is bounded by $4\sqrt{s(1-s)}\,\|[F,E]\|$ and therefore vanishes as $s\to 0$, whereas the normalised drift $|\Delta p_F(E)|$ can grow like $1/\sqrt{s}$ in the same limit. Route C therefore provides its cleanest signatures in moderate-$s$ regimes (e.g.\ $s\gtrsim 0.1$), while rare-event conditioning is better accessed via Routes A or B.
\end{quote}

\textbf{Implementation imperfection.} If the implemented channel is $\delta$-close (in diamond norm) to the ideal unitary conjugation by $U_F$, then the induced systematic error in the Route C contrast is at most $2\delta$. For a meaningful test, one requires this bias to be small compared to the expected contrast scale $4\sqrt{s(1-s)}\,\|[F,E]\|$.

\subsection{Extension to General Readouts}

The $\beta$-bound extends from effects to bounded Hermitian readouts. For a readout $R$ with $\|R\| \le 1$, define the expectation drift
\[
\Delta_F(R) = \frac{\mathrm{Tr}(\rho FR) + \mathrm{Tr}(\rho RF)}{s} - \frac{2\mathrm{Tr}(\rho FRF)}{s}.
\]
Then
\[
|\Delta_F(R)| \le 2\sqrt{\frac{1-s}{s}} \cdot \|[F,R]\|.
\]
The proof is identical: nowhere does the argument use $0 \le E \le I$, only boundedness. For effects, the drift has a probability interpretation; for general readouts, it controls expectation-value discrepancies.

This extension is useful when the natural observable is not positive, such as the coherence readouts $R_{ij} = |i\rangle\langle j| + |j\rangle\langle i|$ appearing in Section~6.

\section{Record Fidelity and Symmetrisation}

\subsection{Setup}

We now turn from gating to symmetrisation. The question is: when we impose a symmetry on a gated state, how much information is lost from the perspective of a given readout?

We treat records as readouts. A readout may be an effect ($0 \le E \le I$) when we discuss probabilities, or a bounded Hermitian observable $R$ with $\|R\| \le 1$ when we discuss coherence-sensitive expectation values. The record fidelity gap compares the readout before and after symmetrisation; it is nonzero precisely when the readout is sensitive to degrees of freedom removed by the twirl.

Throughout this section, we assume the gate respects the symmetry generating the twirl:
\[
[U_g, F] = 0 \quad \text{for all } g \in G,
\]
or equivalently $[H, F] = 0$ in the $U(1)$ case. This ensures that gating and twirling are compatible operations.

When $[H, F] = 0$, twirling and gating commute in the sense that $\mathcal{T}[F\rho F] = F \mathcal{T}[\rho] F$, so the record fidelity gap reflects symmetry-removal in the readout, not an ordering artefact.

\subsection{The Record Fidelity Gap}

For a gated state $\rho_F = F\rho F/s$ and a readout $R$, the record fidelity gap is
\[
\Delta_\mathcal{T}(\rho_F, R) = \mathrm{Tr}(\rho_F R) - \mathrm{Tr}(\mathcal{T}[\rho_F] R).
\]
Using the duality $\mathrm{Tr}(\mathcal{T}[\rho] R) = \mathrm{Tr}(\rho \mathcal{T}^*(R))$, this is equivalently
\[
\Delta_\mathcal{T}(\rho_F, R) = \mathrm{Tr}(\rho_F (R - \mathcal{T}^*(R))).
\]
The gap measures the expectation-value difference between the gated state and its symmetrised version. It quantifies how much the readout $R$ can distinguish $\rho_F$ from $\mathcal{T}[\rho_F]$.

\textbf{Interpretation.} The twirl $\mathcal{T}$ removes degrees of freedom that transform nontrivially under the symmetry $G$. The gap $\Delta_\mathcal{T}(\rho_F, R)$ is large when two conditions hold: the gated state $\rho_F$ has weight on those degrees of freedom, and the readout $R$ is sensitive to them.

\subsection{When the Gap Vanishes}

The record fidelity gap vanishes in two natural cases.

\textbf{$G$-invariant readouts.} If $R$ is $G$-invariant, meaning $U_g^\dagger R U_g = R$ for all $g \in G$, then $\mathcal{T}^*(R) = R$ and
\[
\Delta_\mathcal{T}(\rho_F, R) = \mathrm{Tr}(\rho_F (R - R)) = 0
\]
for all $\rho_F$. The readout cannot see degrees of freedom that the twirl removes, so symmetrisation makes no difference from its perspective.

\textbf{Already-symmetrised states.} If $\rho_F$ is $G$-invariant, meaning $\mathcal{T}[\rho_F] = \rho_F$, then
\[
\Delta_\mathcal{T}(\rho_F, R) = \mathrm{Tr}(\rho_F R) - \mathrm{Tr}(\rho_F R) = 0
\]
for all $R$. The state has no asymmetric degrees of freedom to remove.

These vanishing conditions are features, not degeneracies. They tell us that the gap detects a specific kind of structure: the combination of state asymmetry and readout sensitivity to that asymmetry.

\subsection{The $U(1)$ Case: Energy Dephasing}

The most important special case is $U(1)$ symmetry generated by a Hamiltonian $H$ with non-degenerate spectrum $H = \sum_i \omega_i |i\rangle\langle i|$. The twirl becomes dephasing in the energy eigenbasis:
\[
\mathcal{T}[\rho] = \sum_i |i\rangle\langle i| \rho |i\rangle\langle i| = \sum_i \rho_{ii} |i\rangle\langle i|.
\]
For the coherence readout $R_{ij} = |i\rangle\langle j| + |j\rangle\langle i|$ with $i \ne j$ and $\omega_i \ne \omega_j$, the dual twirl gives $\mathcal{T}^*(R_{ij}) = 0$ because $R_{ij}$ transforms nontrivially under phase rotations.

Therefore the record fidelity gap becomes
\[
\Delta_\mathcal{T}(\rho_F, R_{ij}) = \mathrm{Tr}(\rho_F R_{ij}) - 0 = \mathrm{Tr}(\rho_F R_{ij}) = 2 \mathrm{Re}((\rho_F)_{ij}).
\]
The gap directly measures the real part of the off-diagonal coherence in the gated state. This is the quantity that dephasing removes.

\textbf{Degeneracy caveat.} If $\omega_i = \omega_j$ (degenerate energies), then $R_{ij}$ is $U(1)$-invariant and the gap vanishes identically. The $U(1)$ twirl eliminates coherence between distinct energy eigenspaces but preserves coherence within degenerate subspaces; this is why we require $\omega_i \ne \omega_j$ for the gap to be generically nonzero.

\subsection{Connection to the Coherence Witness}

The coherence witness $W(\rho, F)$ and the record fidelity gap $\Delta_\mathcal{T}(\rho_F, R)$ measure complementary aspects of boundary structure.

The witness $W(\rho, F) = \|F\rho(I-F)\|_1$ is a state-side diagnostic: it measures coherence across the gate boundary before any measurement or symmetrisation. It depends on $\rho$ and $F$ but not on any readout.

The gap $\Delta_\mathcal{T}(\rho_F, R)$ is a readout-side diagnostic: it measures how much a specific readout can distinguish the gated state from its symmetrised version. It depends on $\rho_F$, the symmetry $G$, and the readout $R$.

Together, they give a complete picture:
\begin{align*}
W(\rho, F) \text{ large} &\implies \text{state has cross-boundary coherence} \implies \text{$\beta$-bound relevant} \\
\Delta_\mathcal{T}(\rho_F, R) \text{ large} &\implies \text{readout sees asymmetric structure} \implies \text{symmetrisation matters}
\end{align*}

In experimental design, one checks both: the witness to verify that the prepared state is not trivially block-diagonal, and the gap to verify that the chosen readout is sensitive to the degrees of freedom of interest.

\section{Experimental Vignettes}

\subsection{Vignette Design Principles}

Each vignette specifies four ingredients: (1) Preparation $\rho$: the initial quantum state; (2) Gate $F$: a fixed projector defining the admitted sector; (3) Readout $E$ or $R$: an effect or bounded observable; (4) Symmetry $G$: the group whose twirl defines record structure (where applicable).

A vignette is falsifiable if the $\beta$-bound or record fidelity gap yields a quantitative prediction that can be compared against experimental data. The prediction takes the form: measured drift should satisfy the bound, with saturation possible for appropriately chosen parameters.

The design rule is: choose $(F, E, G)$ such that $\|[F, E]\| > 0$ and, if applicable, choose $R$ not $G$-invariant. Then vary preparation $\rho$ to explore the bound across different values of $s$ and $W(\rho, F)$.

\subsection{Vignette A: Hong--Ou--Mandel Interferometry}

\textbf{System.} Two-mode photonic system with modes $a$ and $b$. The relevant Hilbert space is the three-dimensional two-photon subspace spanned by $\{|2,0\rangle, |1,1\rangle, |0,2\rangle\}$.

\textbf{Gate.} Let $U_{\mathrm{BS}}$ denote the balanced beam-splitter unitary. Define $F_{\mathrm{out}}$ as the projector onto the bunched-output subspace:
\[
F_{\mathrm{out}} = |2,0\rangle\langle 2,0| + |0,2\rangle\langle 0,2|.
\]
This is a rank-2 projector on the two-photon subspace. The effective gate in the input-mode basis is $F = U_{\mathrm{BS}}^\dagger F_{\mathrm{out}} U_{\mathrm{BS}}$.

\textbf{Readout.} We consider two readouts to illustrate the bound's behaviour.

\emph{Null control.} The coincidence-detection effect $E_{\mathrm{coinc}} = |1,1\rangle\langle 1,1|$ projects onto the antibunched output. Since $F_{\mathrm{out}} + E_{\mathrm{coinc}} = I$ on the two-photon subspace, we have $[F_{\mathrm{out}}, E_{\mathrm{coinc}}] = 0$. The $\beta$-bound predicts $\Delta p_F(E_{\mathrm{coinc}}) = 0$, which serves as an experimental control verifying the sanity check of \S2.3.

\emph{Noncommuting readout.} Define
\[
|\chi\rangle = (|1,1\rangle + |2,0\rangle)/\sqrt{2}
\]
and let $E_{\mathrm{mix}} = |\chi\rangle\langle\chi|$. This effect does not respect the bunched/coincidence decomposition: it has support on both $F_{\mathrm{out}}$ and its complement. Therefore $[F_{\mathrm{out}}, E_{\mathrm{mix}}] \ne 0$.

Explicitly, in the ordered basis $\{|2,0\rangle, |0,2\rangle, |1,1\rangle\}$:
\[
F_{\mathrm{out}} = \mathrm{diag}(1, 1, 0), \quad E_{\mathrm{mix}} = \frac{1}{2} \begin{pmatrix} 1 & 0 & 1 \\ 0 & 0 & 0 \\ 1 & 0 & 1 \end{pmatrix}.
\]
The off-diagonal block $F_{\mathrm{out}} E_{\mathrm{mix}} (I - F_{\mathrm{out}})$ is nonzero, giving
\[
\|[F_{\mathrm{out}}, E_{\mathrm{mix}}]\| = 1/2.
\]

\textbf{Prediction.} The $\beta$-bound constrains drift for the noncommuting readout:
\[
|\Delta p_F(E_{\mathrm{mix}})| \le 2\sqrt{\frac{1-s}{s}} \cdot \frac{1}{2} = \sqrt{\frac{1-s}{s}}.
\]
For the null control, $\Delta p_F(E_{\mathrm{coinc}}) = 0$ regardless of $\rho$.

\textbf{Experimental test.} Prepare input states with calibrated $s$, measure both $\Delta p_F(E_{\mathrm{coinc}})$ (expected: zero) and $\Delta p_F(E_{\mathrm{mix}})$ (expected: bounded as above). Standard HOM implementations realise the gate via coincidence post-selection and thereby discard the rejected sector, which precludes Route C (gate-reflection) access in \S4.6. In that case one uses Route A (tomography/calibration) or Route B (interferometric access to the relevant Jordan-product term). Measuring $E_{\mathrm{mix}}$ itself requires auxiliary interferometry to coherently recombine bunched and coincidence sectors; this is experimentally demanding but makes the noncommuting-readout test explicit. Route C becomes feasible only in augmented architectures with coherent phase control across bunched/antibunched sectors.

The Hong--Ou--Mandel effect \citep{HongOuMandel1987} provides a well-characterised platform for such tests, with photon indistinguishability serving as the control parameter for $s$.

\subsection{Vignette B: Atomic Energy-Basis Dephasing}

\textbf{System.} Multi-level atom with Hamiltonian $H = \sum_i \omega_i |i\rangle\langle i|$, where the $\omega_i$ are non-degenerate.

\textbf{Gate.} Projector $F$ onto a subspace spanned by a subset of energy eigenstates, e.g., $F = |1\rangle\langle 1| + |2\rangle\langle 2|$ for a three-level system. The gate commutes with $H$ by construction: $[H, F] = 0$.

\textbf{Symmetry.} $U(1)$ generated by $H$. The twirl is dephasing in the energy eigenbasis.

\textbf{Readout.} Coherence observable $R_{12} = |1\rangle\langle 2| + |2\rangle\langle 1|$. Since $\omega_1 \ne \omega_2$, this is not $U(1)$-invariant.

\textbf{Preparation.} Superposition states $\rho = |\psi\rangle\langle\psi|$ with $|\psi\rangle = c_1|1\rangle + c_2|2\rangle + c_3|3\rangle$. The gate-passage probability is $s = |c_1|^2 + |c_2|^2$. The gated state is
\[
\rho_F = \frac{(c_1|1\rangle + c_2|2\rangle)(c_1^*\langle 1| + c_2^*\langle 2|)}{s}.
\]

\textbf{Prediction.} The record fidelity gap is
\[
\Delta_\mathcal{T}(\rho_F, R_{12}) = 2 \mathrm{Re}((\rho_F)_{12}) = \frac{2 \mathrm{Re}(c_1 c_2^*)}{s}.
\]
This is directly measurable via Ramsey interferometry \citep{Ramsey1950} or related coherence-detection protocols.

\textbf{Experimental test.} Prepare states with known coefficients $c_i$, apply the gate (post-select on the $\{|1\rangle, |2\rangle\}$ subspace), measure the coherence via Ramsey fringes, and compare to the predicted gap. Varying $c_3$ changes $s$ while keeping the within-gate coherence $c_1 c_2^*$ fixed (up to normalisation), allowing exploration of the $s$-dependence.

\subsection{Vignette C: Decoherence-Induced Classicality}

\textbf{System.} A system $S$ coupled to an environment $E$, with total Hamiltonian
\[
H_{\mathrm{tot}} = H_S \otimes I_E + I_S \otimes H_E + H_{\mathrm{int}}.
\]
The interaction Hamiltonian $H_{\mathrm{int}}$ defines a pointer basis $\{|p_i\rangle\}$ for $S$ via einselection \citep{Zurek2003}.

\textbf{Gate.} Projector $F = \sum_{i \in A} |p_i\rangle\langle p_i|$ onto a subset $A$ of pointer states.

\textbf{Interaction Hamiltonian.} Pure-dephasing coupling:
\[
H_{\mathrm{int}} = \sigma_z \otimes B
\]
where $B$ is a bath operator. This interaction picks out $\sigma_z$ eigenstates $\{|0\rangle, |1\rangle\}$ as pointer states, ensuring $[H_{\mathrm{int}}, F \otimes I_E] = 0$ for pointer-basis projectors $F$.

\textbf{Symmetry.} Tracing over the environment produces a CPTP channel $\mathcal{E}_t$ on $S$. For Markovian pure-dephasing dynamics generated by $H_{\mathrm{int}}=\sigma_z\otimes B$, the long-time limit
\[
\mathcal{T}_\infty := \lim_{t\to\infty} \mathcal{E}_t
\]
exists and coincides with full dephasing in the $\sigma_z$ pointer basis, equivalently the $U(1)$ twirl generated by $\sigma_z$. In the record-fidelity formalism we therefore take $\mathcal{T}_\infty$ as the relevant twirl map, and treat finite-time decoherence as an interpolation $\mathcal{E}_t \to \mathcal{T}_\infty$.

\textbf{Readout.} Pointer-basis coherence observable $R_{ij} = |p_i\rangle\langle p_j| + |p_j\rangle\langle p_i|$ for distinct pointer states $p_i, p_j \in A$.

\textbf{Prediction.} Since $\mathcal{T}_\infty$ removes pointer-basis coherences, $\mathrm{Tr}(\mathcal{T}_\infty[\rho_F] R_{ij})=0$ for $i\neq j$, and the associated record fidelity gap is
\[
\Delta_{\mathcal{T}_\infty}(\rho_F, R_{ij}) = \mathrm{Tr}(\rho_F R_{ij}) = 2\mathrm{Re}((\rho_F)_{ij}).
\]
At finite time, the experimentally observed record signal is
\[
\mathrm{Tr}(\mathcal{E}_t[\rho_F] R_{ij}) = 2 \mathrm{Re}((\rho_F(0))_{ij}) \cdot \gamma(t),
\]
where $\gamma(t)$ is the decoherence function (typically exponential decay). The difference $\mathrm{Tr}(\mathcal{E}_t[\rho_F] R_{ij})-\mathrm{Tr}(\mathcal{T}_\infty[\rho_F] R_{ij})$ tracks the decoherence timescale and vanishes as $t \to \infty$.

\textbf{Connection to einselection.} This vignette connects the record fidelity framework to Zurek's environment-induced superselection \citep{Zurek2003}. The pointer states are precisely those for which $W(\rho, F) \to 0$ under decoherence; the record signal above measures the rate of approach to the pointer-basis limit.

\textbf{Accessibility.} Superconducting qubits \citep{DevoretSchoelkopf2013}, trapped ions \citep{LeibfriedEtAl2003}, or NV centres \citep{DohertyEtAl2013} with controlled environment coupling. Decoherence timescales are well-characterised in these systems, and pointer-basis coherences can be measured via quantum state tomography.

\section{Discussion}

\subsection{Relation to Existing Frameworks}

The mathematical ingredients of this paper (projective conditioning, pinching maps, trace-norm bounds) are individually well-known in operator algebra and quantum information theory \citep{NielsenChuang2000,Watrous2018}. What we claim is new is the packaging: an operational drift bound that explicitly connects state-side coherence (the witness $W$), observable-side boundary-crossing ($\|[F,E]\|$), and normalisation cost ($\sqrt{(1-s)/s}$) into a single falsifiable inequality. The vignette design rule (specify $(F, E, G)$ such that $\|[F,E]\| > 0$ and $W(\rho,F) > 0$, then check the bound) provides a template for experimental tests that, to our knowledge, has not been articulated in this form.

We now situate the framework relative to three established approaches.

\textbf{Decoherence program.} The decoherence approach explains the emergence of classical behaviour through environment-induced superselection: pointer states survive because they are robust under system--environment interaction, while superpositions decohere \citep{Zurek2003,Schlosshauer2007,JoosEtAl2003}. Vignette C shows how the record fidelity framework connects to this picture. The asymptotic projection $\mathcal{T}_\infty$ formalises the pointer-basis limit, and the finite-time record signal $\mathrm{Tr}(\mathcal{E}_t[\rho_F]R)$ tracks the decoherence timescale for coherence-sensitive readouts. The $\beta$-bound provides a complementary constraint: even before decoherence completes, the drift in probability assignments is bounded by the gate--observable commutator.

The framework does not replace decoherence theory but provides quantitative bridges. Where decoherence theory explains why certain bases are selected, the $\beta$-bound constrains what happens during the selection process.

\textbf{Consistent histories.} In the consistent histories formalism \citep{Griffiths2002,GellMannHartle1993,Omnes1994}, one defines ``realms'' of histories that satisfy decoherence conditions, and probabilities are assigned only within a realm. The gate $F$ can be viewed as selecting a realm: conditioning on $F$ restricts attention to histories passing through the admitted sector. The $\beta$-bound then constrains how probability assignments behave when the readout $E$ does not fully respect the realm structure (i.e., when $[F,E] \ne 0$).

This connection is suggestive but imperfect. Consistent histories involves sequential projectors and decoherence functionals; our framework treats a single gate--measurement pair. A full translation would require extending the $\beta$-bound to sequential gating, which we leave to future work.

\textbf{Operational and QBist approaches.} In operational approaches, probabilities are understood as encoding an agent's expectations given their measurement context \citep{FuchsSchack2013,Caves2002}. The drift $\Delta p_F(E)$ can be read as the discrepancy between two natural ways of assigning probabilities: the fully-gated assignment $p_\rho(E)$, and the average of the partially-gated assignments $q_L(E)$ and $q_R(E)$. The $\beta$-bound constrains this discrepancy as a function of the agent's choices (gate $F$, readout $E$) and the prepared state $\rho$.

From a QBist perspective, the bound might be interpreted as a coherence constraint on rational probability revision: if an agent conditions on passing a gate, their revised probabilities cannot deviate arbitrarily from certain natural reference points. We do not insist on this interpretation, but note that the framework is compatible with it.

\subsubsection{Relation to Contextuality and the Quantum Zeno Effect}

While the $\beta$-bound constrains probabilities involving non-commuting operations, it differs in kind from contextuality inequalities and from Zeno dynamics. Contextuality results typically establish (often in a state-independent way, or by optimisation over states) the impossibility of context-independent value assignments across incompatible measurement contexts \citep{KochenSpecker1967}. By contrast, the $\beta$-bound is a state-dependent stability inequality: it bounds a specific operational functional $\Delta p_F(E)$ for fixed $(\rho,F,E)$, without claiming that any alternative assignment scheme is impossible. Similarly, the quantum Zeno effect concerns repeated measurements and dynamical backaction under frequent monitoring \citep{MisraSudarshan1977}. The present framework addresses single-step gating mismatch and the resulting conditioned-probability drift, independent of subsequent evolution.

\subsection{What the Framework Does Not Claim}

We emphasise several limitations.

\textbf{Not a solution to the measurement problem.} The $\beta$-bound constrains probability assignments under gating but does not explain why measurements have definite outcomes. The framework is compatible with collapse, many-worlds, and hidden-variable interpretations; it provides structure that any interpretation must accommodate, but does not adjudicate between them.

\textbf{Not a claim about ontology.} We do not claim that the gate $F$ corresponds to a physical process, an observer's knowledge, or an objective feature of reality. The framework is operational: it specifies mathematical relationships between preparations, gates, and readouts, and yields testable predictions. Ontological interpretation is left to the reader.

\textbf{Not a complete theory of records.} The record fidelity gap captures one aspect of record structure, namely sensitivity to symmetrisation, but does not address questions about what constitutes a record, how records are created, or why records are stable. These are important questions that go beyond the present framework.

\subsection{Limitations and Open Questions}

Several directions remain open.

\textbf{Projector-only formulation.} The present paper treats gates as projectors ($F^2 = F$). A natural generalisation uses positive contractions $0 \le F \le I$, representing ``soft'' gates. The proof strategy extends if we formulate the drift and commutator in terms of $F^{1/2}$, but controlling $\|[F^{1/2},E]\|$ by $\|[F,E]\|$ requires a spectral-gap condition on $F$ (or a bounded-inverse condition on its support, restricted to the prepared ensemble). We leave the soft-gate generalisation to future work.

\textbf{Sequential gating.} The $\beta$-bound treats a single gate. Sequential gating (conditioning on passing $F_1$, then $F_2$, and so on) appears in consistent histories and quantum trajectories \citep{Griffiths2002,WisemanMilburn2009}. Extending the bound to sequences would require tracking how drift accumulates or cancels across multiple gates.

\textbf{Infinite-dimensional systems.} The present results assume finite dimension or bounded operators on a trace-class state. Unbounded observables (such as position or momentum) require additional care with domains \citep{ReedSimon1980}. The core inequalities should extend with standard modifications, but we have not verified the details.

\textbf{Thermodynamic constraints.} The coherence witness $W(\rho,F)$ and the record fidelity gap $\Delta_\mathcal{T}(\rho_F, R)$ both involve coherence across boundaries. Coherence is a thermodynamic resource \citep{LostaglioeEtAl2015,StreltsovEtAl2017}, and one expects connections to fluctuation theorems and entropy production. Exploring these connections could yield constraints on the energetic cost of maintaining or erasing records.

\section{Conclusion}

We have developed a measurement-theoretic framework for projective gating, centred on three results.

The $\beta$-bound (Theorem 4.1) controls the drift in probability assignments when gating and measurement fail to commute. The bound
\[
|\Delta p_F(E)| \le 2\sqrt{\frac{1-s}{s}} \cdot \|[F,E]\|
\]
makes precise the intuition that boundary-crossing effects are controlled by the commutator, with amplification when conditioning on rare events.

The coherence witness $W(\rho,F) = \|F\rho(I-F)\|_1$ diagnoses state-side coherence across the gate boundary, satisfying the state budget inequality $W(\rho,F) \le \sqrt{s(1-s)}$.

The record fidelity gap $\Delta_\mathcal{T}(\rho_F, R)$ measures how much expectation values change under symmetrisation, vanishing precisely when the readout cannot see the degrees of freedom that the twirl removes.

Together, these quantities form a unified package. The drift is controlled by state-side coherence (the witness), observable-side boundary-crossing (the commutator), and normalisation cost (the $\sqrt{(1-s)/s}$ factor). The vignettes of Section~6 show how to instantiate this package in concrete experimental settings.

The framework is operationally accessible via multiple routes with platform-dependent trade-offs, and interpretation-neutral. It does not solve the measurement problem or make ontological claims. What it provides is a set of quantitative constraints that any account of quantum measurement must respect, along with a template for experimental tests.

We hope these tools prove useful for both foundational research and practical applications in quantum information, metrology, and the study of decoherence.

\appendix

\section{Proofs}

\subsection{Proof of Lemma 3.1 (State Budget Inequality)}

\begin{lemma}
For any density operator $\rho$ and projector $F$ with $s = \mathrm{Tr}(\rho F) \in (0,1]$,
\[
W(\rho,F) = \|F\rho(I-F)\|_1 \le \sqrt{s(1-s)}.
\]
\end{lemma}

\begin{proof}
The trace norm is characterised by
\[
\|X\|_1 = \sup \{ |\mathrm{Tr}(U^\dagger X)| : \|U\| \le 1 \}.
\]
Let $U$ be any operator with $\|U\| \le 1$. We compute
\[
|\mathrm{Tr}(U^\dagger F\rho(I-F))| = |\mathrm{Tr}((\rho^{1/2}F)^\dagger (\rho^{1/2}(I-F)U))|.
\]
By the Cauchy--Schwarz inequality for the Hilbert--Schmidt inner product $\langle A,B\rangle = \mathrm{Tr}(A^\dagger B)$,
\[
|\mathrm{Tr}(A^\dagger B)| \le \|A\|_2 \|B\|_2
\]
where $\|A\|_2 = \sqrt{\mathrm{Tr}(A^\dagger A)}$ is the Hilbert--Schmidt norm.

Applying this with $A = \rho^{1/2}F$ and $B = \rho^{1/2}(I-F)U$:
\begin{align*}
\|\rho^{1/2}F\|_2^2 &= \mathrm{Tr}(F\rho^{1/2}\rho^{1/2}F) = \mathrm{Tr}(F\rho F) = \mathrm{Tr}(\rho F^2) = \mathrm{Tr}(\rho F) = s, \\
\|\rho^{1/2}(I-F)U\|_2^2 &= \mathrm{Tr}(U^\dagger(I-F)\rho(I-F)U) \le \|U\|^2 \mathrm{Tr}((I-F)\rho(I-F)) \\
&= \mathrm{Tr}(\rho(I-F)^2) = \mathrm{Tr}(\rho(I-F)) = 1-s,
\end{align*}
where we used $\|U\| \le 1$ and cyclicity of trace.

Therefore
\[
|\mathrm{Tr}(U^\dagger F\rho(I-F))| \le \sqrt{s} \cdot \sqrt{1-s} = \sqrt{s(1-s)}.
\]
Taking the supremum over $U$ with $\|U\| \le 1$ gives the result.
\end{proof}

\subsection{Proof of Theorem 4.1 (\texorpdfstring{$\beta$}{β}-Bound)}

\begin{theorem}
Let $\rho$ be a density operator, $F$ a projector, and $E$ an effect. Define $s = \mathrm{Tr}(\rho F) \in (0,1]$, $p_\rho(E) = \mathrm{Tr}(\rho FEF)/s$, and $\varepsilon = \|[F,E]\|$. Define the symmetric partial-gating drift
\[
\Delta p_F(E) = \frac{\mathrm{Tr}(\rho FE) + \mathrm{Tr}(\rho EF)}{s} - 2p_\rho(E).
\]
Then
\[
|\Delta p_F(E)| \le 2\sqrt{\frac{1-s}{s}} \cdot \varepsilon.
\]
\end{theorem}

\begin{proof}
\textbf{Step 1 (Block form and commutator size).} Fix the orthogonal decomposition $\mathcal{H} = \mathrm{ran}(F) \oplus \mathrm{ran}(I-F)$. Write operators in $2 \times 2$ block form:
\[
E = \begin{pmatrix} A & B \\ B^\dagger & C \end{pmatrix}, \quad \rho = \begin{pmatrix} \rho_{11} & \rho_{12} \\ \rho_{21} & \rho_{22} \end{pmatrix}
\]
where $A = FEF$, $B = FE(I-F)$, $C = (I-F)E(I-F)$, and similarly $\rho_{11} = F\rho F$, $\rho_{12} = F\rho(I-F)$, $\rho_{21} = (I-F)\rho F$, $\rho_{22} = (I-F)\rho(I-F)$.

For a projector $F$,
\[
[F,E] = F E - E F = \begin{pmatrix} 0 & B \\ -B^\dagger & 0 \end{pmatrix}.
\]
Hence $\|[F,E]\| = \|B\|$. Therefore $\varepsilon = \|B\|$.

\textbf{Step 2 (Express the drift by off-diagonal blocks).} We compute the relevant traces:
\[
\mathrm{Tr}(\rho FEF) = \mathrm{Tr}(\rho_{11} A).
\]
For the one-sided terms:
\begin{align*}
\mathrm{Tr}(\rho FE) &= \mathrm{Tr}(\rho_{11} A) + \mathrm{Tr}(\rho_{21} B), \\
\mathrm{Tr}(\rho EF) &= \mathrm{Tr}(\rho_{11} A) + \mathrm{Tr}(\rho_{12} B^\dagger).
\end{align*}
Therefore the drift is
\[
\Delta p_F(E) = \frac{\mathrm{Tr}(\rho_{21} B) + \mathrm{Tr}(\rho_{12} B^\dagger)}{s}.
\]

\textbf{Step 3 (Trace-norm bound).} Using $|\mathrm{Tr}(XY)| \le \|X\|_1 \|Y\|$ for trace-class $X$ and bounded $Y$:
\[
|\mathrm{Tr}(\rho_{21} B)| \le \|\rho_{21}\|_1 \|B\|, \quad |\mathrm{Tr}(\rho_{12} B^\dagger)| \le \|\rho_{12}\|_1 \|B^\dagger\| = \|\rho_{12}\|_1 \|B\|.
\]
Hence
\[
|\Delta p_F(E)| \le \frac{(\|\rho_{21}\|_1 + \|\rho_{12}\|_1) \|B\|}{s}.
\]

\textbf{Step 4 (Bound the off-diagonal trace norms by $s$).} By Lemma 3.1,
\[
\|\rho_{12}\|_1 = \|F\rho(I-F)\|_1 \le \sqrt{s(1-s)}, \quad \|\rho_{21}\|_1 = \|(I-F)\rho F\|_1 \le \sqrt{s(1-s)}.
\]

\textbf{Step 5 (Combine).} Substituting into the bound from Step 3:
\[
|\Delta p_F(E)| \le \frac{2\sqrt{s(1-s)} \cdot \|B\|}{s} = 2\sqrt{\frac{s(1-s)}{s^2}} \cdot \|B\| = 2\sqrt{\frac{1-s}{s}} \cdot \|B\|.
\]
Using $\|B\| = \varepsilon$ from Step 1:
\[
|\Delta p_F(E)| \le 2\sqrt{\frac{1-s}{s}} \cdot \varepsilon. \qedhere
\]
\end{proof}

\textbf{Saturation criterion.} The bound is saturated when two conditions hold: (i) $\rho$ is pure with support in both sectors such that $\|\rho_{12}\|_1 = \sqrt{s(1-s)}$ (maximal cross-boundary coherence at fixed $s$); and (ii) the off-diagonal block $B = FE(I-F)$ is aligned with the polar decomposition of $\rho_{12}$, meaning $\mathrm{Tr}(\rho_{12} B^\dagger) = \|\rho_{12}\|_1 \|B\|$. In general, misalignment between $\rho_{12}$ and $B$, or mixedness in $\rho$, causes the bound to be strict.

\subsection{Proof of Twirl--Gate Commutation}

\begin{proposition}
Let $G$ be a compact group with unitary representation $g \mapsto U_g$, and let $\mathcal{T}$ be the associated twirl. Let $F$ be a projector satisfying $U_g F U_g^\dagger = F$ for all $g \in G$. Then for any operator $X$,
\[
\mathcal{T}[FXF] = F \mathcal{T}[X] F.
\]
\end{proposition}

\begin{proof}
By definition,
\[
\mathcal{T}[FXF] = \int_G U_g (FXF) U_g^\dagger \, d\mu(g).
\]
Using $U_g F = F U_g$ (which follows from $U_g F U_g^\dagger = F$):
\[
U_g (FXF) U_g^\dagger = U_g F X F U_g^\dagger = F U_g X U_g^\dagger F.
\]
Therefore
\[
\mathcal{T}[FXF] = \int_G F U_g X U_g^\dagger F \, d\mu(g) = F \left(\int_G U_g X U_g^\dagger \, d\mu(g)\right) F = F \mathcal{T}[X] F. \qedhere
\]
\end{proof}

\textbf{Corollary.} Under the same hypotheses, for $\rho_F = F\rho F/s$ with $s = \mathrm{Tr}(\rho F)$,
\[
\mathcal{T}[\rho_F] = \frac{F \mathcal{T}[\rho] F}{s}.
\]
This shows that twirling and gating commute: one obtains the same result by twirling first and then gating, or by gating first and then twirling.

\subsection{Gate-reflection Identity}

\begin{proposition}
Let $F$ be a projector and define the unitary reflection $U_F=2F-I$. For any density operator $\rho$ and bounded operator $E$,
\[
\Delta p_F(E)=\frac{\mathrm{Tr}(\rho E)-\mathrm{Tr}(\rho U_F E U_F)}{2s},
\quad s=\mathrm{Tr}(\rho F),
\]
where $\Delta p_F(E)=\big(\mathrm{Tr}(\rho FE)+\mathrm{Tr}(\rho EF)-2\mathrm{Tr}(\rho FEF)\big)/s$.
\end{proposition}

\begin{proof}
Expand
\[
U_F E U_F=(2F-I)E(2F-I)=4FEF-2FE-2EF+E.
\]
Taking $\mathrm{Tr}(\rho\,\cdot)$ gives
\[
\mathrm{Tr}(\rho U_F E U_F)=4\mathrm{Tr}(\rho FEF)-2\mathrm{Tr}(\rho FE)-2\mathrm{Tr}(\rho EF)+\mathrm{Tr}(\rho E).
\]
Rearranging,
\[
\mathrm{Tr}(\rho FE)+\mathrm{Tr}(\rho EF)=\frac{4\mathrm{Tr}(\rho FEF)+\mathrm{Tr}(\rho E)-\mathrm{Tr}(\rho U_F E U_F)}{2}.
\]
Substituting into the definition of $\Delta p_F(E)$ yields
\[
\Delta p_F(E)=\frac{\mathrm{Tr}(\rho E)-\mathrm{Tr}(\rho U_F E U_F)}{2s}. \qedhere
\]
\end{proof}

\section{Vignette Technical Details}

The parameters below are illustrative and reflect typical orders of magnitude for current implementations. Specific values vary by platform and should be calibrated experimentally.

\subsection{Optical Setup Parameters (Vignette A)}

\textbf{Beam-splitter specification.} A balanced 50:50 beam splitter with unitary
\[
U_{\mathrm{BS}} = \exp(i\pi(a^\dagger b + b^\dagger a)/4)
\]
acting on mode operators $a$, $b$. On the two-photon subspace spanned by $\{|2,0\rangle, |1,1\rangle, |0,2\rangle\}$, the beam splitter acts as:
\begin{align*}
U_{\mathrm{BS}} |1,1\rangle &= (|2,0\rangle + |0,2\rangle)/\sqrt{2} \quad (\text{bunched}) \\
U_{\mathrm{BS}} |2,0\rangle &= (|2,0\rangle - 2|1,1\rangle + |0,2\rangle)/2 \\
U_{\mathrm{BS}} |0,2\rangle &= (|2,0\rangle + 2|1,1\rangle + |0,2\rangle)/2
\end{align*}

\textbf{Gate projector.} The bunched-output projector in the output basis is
\[
F_{\mathrm{out}} = |2,0\rangle\langle 2,0| + |0,2\rangle\langle 0,2|.
\]
In the input basis, $F = U_{\mathrm{BS}}^\dagger F_{\mathrm{out}} U_{\mathrm{BS}}$.

\textbf{Coincidence windows.} For typical parametric down-conversion sources \citep{KwiatEtAl1995}: photon arrival jitter $\sim$100~ps; coincidence window 1--10~ns (platform dependent); expected coincidence rate (distinguishable) $10^3$--$10^5$/s; expected coincidence rate (indistinguishable) suppressed by HOM dip.

\textbf{Bunching probability calibration.} The gate-passage probability $s = \mathrm{Tr}(\rho F_{\mathrm{out}})$ is measured directly as the bunching fraction: the ratio of bunched detections to total two-photon detections.

\subsection{Atomic System Parameters (Vignette B)}

\textbf{Level structure.} Consider alkali atoms (e.g., $^{87}$Rb) with: $|1\rangle$, $|2\rangle$: ground-state hyperfine levels, splitting $\sim$GHz; $|3\rangle$: excited state, splitting $\sim$THz from ground states.

\textbf{State preparation.} Coherent superpositions prepared via: optical pumping to initialise in $|1\rangle$; microwave or Raman pulses to create superpositions; optional coupling to $|3\rangle$ via resonant laser.

\textbf{Gate implementation.} Post-selection on the $\{|1\rangle, |2\rangle\}$ subspace via: state-selective fluorescence detection; push-out pulses removing population in $|3\rangle$; measurement-based heralding.

\textbf{Coherence measurement.} Ramsey interferometry \citep{Ramsey1950}: first $\pi/2$ pulse creates superposition; free evolution for time $\tau$; second $\pi/2$ pulse with variable phase $\phi$; population measurement gives $P(|2\rangle) = \frac{1}{2} + \frac{1}{2} \mathrm{Re}((\rho_F)_{12} e^{i\phi})$.

Fringe visibility directly yields $|(\rho_F)_{12}|$; phase of fringes yields $\arg((\rho_F)_{12})$.

\textbf{Decoherence rates.} Typical $T_2$ times vary widely: magnetically shielded $T_2 \sim 0.1$--1~s; unshielded $T_2 \sim 1$--100~ms; environment-limited $T_2 \sim 10$--100~$\mu$s.

\subsection{System--Environment Parameters (Vignette C)}

\textbf{Model system.} Superconducting transmon qubit or spin qubit with controlled dephasing \citep{DevoretSchoelkopf2013}.

\textbf{Pointer states.} Computational basis states $|0\rangle$, $|1\rangle$, selected by the $\sigma_z \otimes B$ interaction structure.

\textbf{Interaction Hamiltonian.} Pure-dephasing coupling:
\[
H_{\mathrm{int}} = \sigma_z \otimes B
\]
where $B$ is a bath operator. This picks out $\sigma_z$ eigenstates as pointer states.

\textbf{Decoherence function.} For a thermal or Gaussian bath:
\[
\gamma(t) = \exp(-t/T_2) \quad \text{or} \quad \exp(-(t/T_\phi)^2)
\]
where $T_2$ is the coherence time and $T_\phi$ characterises Gaussian decay.

\textbf{Typical parameters (order of magnitude):} superconducting qubits $T_2 \sim 10$--100~$\mu$s; spin qubits $T_2 \sim 1$--1000~$\mu$s; NV centres $T_2 \sim 1$--10~ms; measurement fidelity 90--99\%.

\textbf{Tomography protocol.} Quantum state tomography via: dispersive or projective readout; multiple measurement bases; maximum likelihood estimation with readout correction.

\subsection{Error Budget Analysis}

\textbf{Sources of systematic error:}

\emph{1. Gate imperfection.} If $F_{\mathrm{actual}} \ne F_{\mathrm{ideal}}$, the measured drift includes contributions from gate error as well as the $\beta$-bound effect. Mitigation: characterise gate via process tomography; bound gate error contribution separately.

\emph{2. Readout crosstalk.} If the readout $E_{\mathrm{actual}}$ has unintended components, $\|[F, E_{\mathrm{actual}}]\|$ differs from the intended value. Mitigation: calibrate readout POVM; compute commutator for actual readout.

\emph{3. State preparation error.} If $\rho_{\mathrm{actual}} \ne \rho_{\mathrm{intended}}$, the values of $s$ and $W(\rho,F)$ are affected. Mitigation: state tomography before gating; propagate preparation uncertainty to bound predictions.

\emph{4. Statistical uncertainty.} Finite sample sizes introduce uncertainty in estimated probabilities. For $N$ measurement shots:
\[
\sigma(p) \approx \sqrt{p(1-p)/N}, \quad \sigma(\Delta p_F) \approx \sqrt{\sigma(q_L)^2 + \sigma(q_R)^2 + 4\sigma(p_\rho)^2}.
\]
Typical requirements: $N \sim 10^4$--$10^6$ shots for sub-percent statistical uncertainty.

\emph{5. Drift and instability.} Slow drifts in experimental parameters can mimic or mask the $\beta$-bound effect. Mitigation: interleaved measurements; real-time calibration; Allan variance analysis.

\emph{6. Normalisation instability (small-$s$ regimes).} Any estimator that normalises by $\hat{s}$ inherits amplified uncertainty when $s\ll 1$. For Route C in \S4.6 with $\widehat{\Delta p_F(E)}=(\hat{p}_0-\hat{p}_1)/(2\hat{s})$, shot-noise scaling yields uncertainty growing approximately as $1/(s\sqrt{N})$ (and worse if uncertainty in $\hat{s}$ dominates). Maintaining fixed relative precision in rare-event regimes therefore requires $N$ scaling steeply with $1/s^2$.

\textbf{Falsification criterion.} The $\beta$-bound is falsified if:
\[
|\Delta p_F(E)|_{\mathrm{measured}} - 2\sigma_{\mathrm{total}} > 2\sqrt{\frac{1-s}{s}} \cdot \|[F,E]\|
\]
where $\sigma_{\mathrm{total}}$ includes both statistical and systematic uncertainties. This is a conservative criterion that accounts for the one-sided nature of the bound.

\bibliographystyle{plainnat}

\end{document}